\documentclass[11pt]{article}
\usepackage[utf8]{inputenc}
\usepackage{amsmath}
\usepackage{amsthm}
\usepackage{amssymb}
\usepackage{amsfonts}
\usepackage{geometry}
\usepackage{xcolor}
\usepackage{multirow}
\usepackage[square, numbers, comma, sort & compress]{natbib}  
\theoremstyle{plain}\newtheorem{theorem}{Theorem}[section]
\theoremstyle{plain}\newtheorem{proposition}[theorem]{Proposition}
\theoremstyle{plain}\newtheorem{lemma}[theorem]{Lemma}
\theoremstyle{plain}\newtheorem{corollary}[theorem]{Corollary}
\theoremstyle{definition}\newtheorem{definition}[theorem]{Definition}
\theoremstyle{remark}\newtheorem{remark}[theorem]{Remark}
\theoremstyle{definition}\newtheorem{example}[theorem]{Example}

\def\m{\mathbb{M}}
\def\k{\mathbb{K}}
\def\l{\mathbb{L}}

\def\z{\mathbb{Z}}
\def\c{\mathbb{C}}
\def\p{\mathbb{P}}
\def\sl{SL_{2}(\mathbb{C})}
\def\a{\alpha}
\def\b{\beta}
\def\g{\gamma}
\def\D{\Delta}
\def\d{\delta}

\title{Liouvillian solutions for second order linear differential 
equations with polynomial coefficients}

\author{Primitivo B. Acosta-Hum\'anez$^{a,b}$, David Bl\'azquez-Sanz$^c$ \& Henock Venegas-G\'omez$^{c}$}

\date{\vspace{-2ex}}

\usepackage{natbib}
\usepackage{graphicx}

\begin{document}

\maketitle
\begin{center}
$^a$ Instituto Superior de Formaci\'on Docente Salom\'e Ure\~na, Santiago, Dominican Republic.\\

$^b$ Fac. Ciencias B\'asicas y Biom\'edicas, Universidad Sim\'on Bol\'ivar, Barranquilla, Colombia.\\

$^c$ Escuela de Matem\'aticas, Facultad de Ciencias, Universidad Nacional de Colombia - Sede Medell\'in, Colombia.
\end{center}
\begin{abstract}
In this paper we present an algebraic study concerning the general second order linear differential equation with polynomial coefficients. By means of Kovacic's algorithm and asymptotic iteration method we find a degree independent algebraic description of the spectral set: the subset, in the parameter space, of Liouville integrable differential equations. For each fixed degree, we prove that the spectral set is a countable union of non accumulating algebraic varieties. This algebraic description of the spectral set allow us to bound the number of eigenvalues for algebraically quasi-solvable potentials in the Schr\"odinger equation.
\medskip

    \medskip
    
\end{abstract}
\section*{Introduction}

Eigenvalue problems and explicit solutions for linear differential equations is a topic that has drawn the attention of many researchers in the last decades. Let us summarize briefly the known results about explicit solutions for the one dimensional stationary Schr\"odinger equation.\\

We start mentioning that Natanzon in 1971, see \cite{N71}, introduced \emph{exactly solvable potentials},  which today are known as \emph{Natanzon potentials}. The seminal work of Natanzon inspired further researchers about exactly solvable potentials, although in the sense of Natanzon exactly solvable potentials also include potentials in where Schr\"odinger equations have eigenfunctions of hypergeometric type, not necessarily Liouvillian functions. The exactly solvable potentials, also known as solvable potentials, we extended to Schr\"odinger equations with explicit eigenfunctions. In this sense, solvable potentials are related to Schr\"odinger equations with eigenfunctions belonging to the set of special functions (Airy, Bessel, Error, Ei, Hypergeometric, Whittaker, Heun), not necessarily Liouvillian! Moreover, in case of Coulomb and 3D harmonic oscillator potentials correspond to Schr\"odinger equations which are transformed into Whittaker differential equations, Martinet-Ramis in \cite{MR89} established the necessary and sufficient conditions to determine the obtaining of Liouvillian solutions of the Whittaker differential equations.\\

To avoid confusion between explicit and Liouvillian solutions it was introduced the concept of \emph{algebraic spectrum} in in \cite{A09}. Also known as Liouvillian spectral set it is the set of eigenvalues for which the Schr\"odinger equation has Liouvillian eigenfunctions, see also \cite{A10,AMW11}. In some scenarios it is known that bounded eigenfunctions of Schr\"odinger operator are necessarily Liouvillian, see \cite{BY12}.  
Potentials with infinite countable algebraic spectrum are called \emph{algebraically solvable potentials} and those with finite algebraic spectrum \emph{algebraically quasi-solvable potentials}, for complete details see \cite[\S 3.1, pp. 316]{AMW11} and see also \cite{A09,A10}.  \\

On the other hand, Turbiner in 1988, see \cite{T88}, following the same philosophy of Natanzon, introduced \emph{quasi-solvable potentials}. The seminal paper of Turbiner leaded to the seminal paper of Bender and Dunne in 1996, see \cite{BD96}, in where they obtain a family of orthogonal polynomials in the energy values of the Schr\"odinger equation with sextic anharmonic potentials, see also \cite{2SHC06} for the study of more general sextic anharmonic oscillators. Due to Schr\"odinger equation with quartic anharmonic oscillator potential falls in \emph{triconfluent Heun equation}, see \cite{DL92}, it is in some sense a generalized Natanzon potential (exactly solvable) although there no exist Liouvillian eigenfunctions. In a similar way for algebraically solvable potentials, in \cite[\S 3.1, pp. 316]{AMW11} also was introduced the concept of algebraically quasi-solvable potential as the potentials with finite non empty algebraic spectrum, see also \cite{A09,A10}. Examples of algebraically solvable potentials and algebraically quasi-solvable potentials (quartic and sextic oscillators) were presented in \cite{A09,A10,AMW11} using \cite[Theorem 2.5, pp. 276]{AB08}, which corresponds to the application of Kovacic algorithm for reduced second order linear differential equation with polynomial coefficients.\\

In a more general sense, the study of second order linear differential equations with polynomial coefficients can be done through Kovacic's algorithm, see \cite{Kov86} and see also \cite{DL92}, as well by Asymptotic Iteration Method., see \cite{CHS}. Recently Combot in \cite{C18} developed another method to obtain exactly solvable potentials, in the sense of Natanzon, involving \emph{rigid functions} in the sense of Katz.\\ %Although the method developed in \cite{C18} is interesting, we will not used here. \\

 From now on, along this paper, we say exactly solvable potential to mean algebraically solvable potential and we say quasi-solvable potential to mean algebraically quasi-solvable potential. That is, along this paper we only consider Liouvillian solutions. Thus, some results of this paper are obtained by the applying of Asymptotic Iteration Method and by applying of Kovacic Algorithm.  It is because that we study parameterized differential equations and then we obtain families of solutions depending on parameters with these methods. Moreover, theoretical results using affine algebraic varieties, in particular spectral varieties, are presented to obtain the conditions for the solutions and eigenvalues of one dimensional Schr\"odinger equations with anharmonic potentials and for more general second order linear differential equations.\\
The structure of the paper is as follows. Section \ref{para} is devoted to the definitions of parameter space $\mathbb P_{2n}$, spectral set $\mathbb L_{2n}$, spectral varieties $\mathbb L_{2n,d}$ and the statement of our first main result, Theorem \ref{th:main_structure}. Section \ref{liou} is devoted to the definition of polynomial-hyperexponential solutions, the reduction of the parameter space through D'Alembert transformation, and Kovacic's algorithm. The analysic of equation $y''=(x^{2n}+\mu x^{n-1})y$ allows us to prove the non-emptyness of the spectral varieties $\mathbb L_{2n,kn}$ and $\mathbb L_{2n,kn+1}$ (Corolary \ref{cor:no_empty}). Section \ref{auxi} contains the results of this paper related to Asymptotic Iteration Method. We find a sequence of differential polynomials $\Delta_d(a,b)$ in two variables that codify the equations of the spectral varieties $\mathbb L_{2n,d}$ inpendently of $n$ (Theorem \ref{seqdelta}). The proof of Theorem \ref{th:main_structure} is included at the end of the section. Section \ref{schr} is devoted to the Liovullian solutions of Schr\"odinger equations with polynomial potentials. We proof that the number of the values of the energy parameter is bounded by \emph{the arithmetic condition} which is a simple function of the coefficients of the potential (Theorem \ref{th:eigen}).

\section{Parameter space}\label{para}

Let us consider the general second order linear differential equation,
\begin{equation}\label{eq:1}
u'' + P(x)u' + Q(x)u = 0,
\end{equation}
with polynomial coefficients $P(x)= \sum_{j=1}^{n} p_jx^j\in\c[x]_{\leq n}$ and $Q(x) = \sum_{j=1}^{2n} q_jx^j \in\c[x]_{\leq 2n}$. We also take into account a non-degeneracy condition $p_n^2-q_{2n}\neq 0$. Thus, the \emph{parameter space} corresponding to the family of equations \eqref{eq:1} is,
$$\p_{2n} = (\c[x]_{\leq n} \times \in\c[x]_{\leq 2n}) - \{p_n^2-q_{2n} = 0\},$$
that we consider an an affine algebraic variety of dimension $3n+2$ with affine coordinates $p_0,\ldots,p_n,q_0,\ldots,q_{2n}$. Our purpose is to describe algebraically the \emph{spectral set} $\l_{2n}\subseteq \p_{2n}$. That is, the set of equations in the family \eqref{eq:1} admitting a Liouvillian solution. An important class of Liouvillian functions, specifically relevant for the integrability of equation \eqref{eq:1} is the following. 

\begin{definition}
A polynomial-hyperexponential function of polynomial degree $d$ and exponential degree $k$ is a function of the form
\begin{equation}\label{sol:liouvillian_1}
u(x) = P_d(x)e^{\int A_k(x)dx},
\end{equation}
with $P_d(x)$ and $A_k(x)$ polynomials of degree $d$ and $k$ respectively. 
\end{definition}

%From some already known analysis of equation \eqref{eq:1}, based in Kovacic's algorithm (Theorem \ref{polychar}) we know that if an equation in $\p_{2n}$ has a Liouvillian solution then it has an polynomial-hyperexponential solution. 

\begin{definition}
The spectral subvariety $\l_{2n,d}$ is the subset of $\l_{2n}$ corresponding to equations in the family \eqref{eq:1} having a polynomial-hyperexponentional solution of polynomial degree $d$.
\end{definition}

%We will see that $\l_{2n,d}$ is an algebraic subvariety of $\p_{2n}$, 
%Therefore, the set $\l_{2n}$ of equations having a Liouvillian solution (or, in this particular case, solvable Galois group) is: % a enumerable union of algebraic subvarieties (Corollary \ref{cor:polychar}),
%$$\l_{2n} = \bigcup_{d=0}^\infty \l_{2n,d}.$$

\begin{theorem}\label{th:main_structure}
Let $\l_{2n}\subset \p_{2n}$ be the set of equations in family \eqref{eq:1} having a Liouvillian solution, and $\l_{2n,d}$ be the set of equations in family \eqref{eq:1} having a polynomial-hyperexponential solution of polynomial degree $d$. The following statements hold:
\begin{itemize}
    \item[(a)] For any fixed $n\in \mathbb N$ there is an infinite set of values of $d$ such that $\l_{2n,d}$ is not empty.
    \item[(b)] If not empty, $\l_{2n,d}$ is an algebraic variety of codimension $\leq n$ in $\p_{2n}$.
    \item[(c)] For any $d \neq k$ the algebraic varieties $\l_{2n,d}$ and $\l_{2n,k}$ are disjoint in $\p_{2n}$.
    \item[(d)] Any compact subset of $\p_{2n}$ intersects only a finite number of algebraic varieties of the family $\{\l_{2n,d}\}_{d\in \mathbb N}$.
\end{itemize}
Furthermore,
$$\l_{2n} = \bigcup_{d=0}^\infty \l_{2n,d}.$$
Therefore we conclude that $\l_{2n}$ is a singular analytic submanifold of $\p_{2n}$ consisting in the enumerable union of pairwise disjoint algebraic varieties of codimension $\leq n$ in $\p_{2n}$.
\end{theorem}

In what follows we will deal with the proof of Theorem \ref{th:main_structure} and the calculation of the equations of the spectral subvarieties $\l_{2n,d}$ in suitable coordinates.

%Then, in the particular case of the $1$-dimensional stationary Schr\"odinger equation with polynomial potential, we explore the elimination of the energy parameter in the equation of $\mathbb L_{2n}$ to classify all possible quasi-solvable polynomial potentials. 

\section{Liouvillian solutions}\label{liou}

\subsection{Reduction of the parameter space}

As it is well known, equation \eqref{eq:1} can be reduced to trace free form
\begin{equation}\label{eq:2}
y'' = R(x)y
\end{equation}
by means of D'Alembert transform $u = {\rm exp}\left(-\frac{1}{2}\int P(x)dx \right)y$, where $R(x) = \frac{P(x)^2}{4} + \frac{P'(x)}{2} - Q(x)$. Note that the degree of $R(x)$
is not greater than ${\rm max}\{ {\rm deg}(Q(x)) ,2{\rm deg} (P(x)) \}$. Note that the family of equations of the form \eqref{eq:3} with $R(x)$ of fixed degree $2n$ are parameterized by the space $\k_{2n} =\mathbb C[x]_{2n}$ of polynomials of degree $2n$ that we see as an affine algebraic variety of dimension $2n+1$, parameterized by the coefficients of $R(x)$ and thus isomorphic to $\c^*\times \mathbb \c^{2n}$. Note that the family \eqref{eq:2} is included in \eqref{eq:1}, where $R(x)\in \mathbb K_{2n}$ corresponds to $(0,-R(x))\in \mathbb P_{2n}$. The D'Alembert transformation is a polynomial map in the coefficients of $P(x)$ and $Q(x)$ and it can be seen as a retract,
$${\rm dal}_{2n}\colon \mathbb P_{2n}\to \mathbb K_{2n}\subset \p_{2n}, \quad (P(x),Q(x))\mapsto R(x)= \frac{P(x)^2}{4} + \frac{P'(x)}{2} - Q(x)\mapsto (0,-R(x))$$
of the natural inclusion $\k_{2n} \subset \p_{2n}$. Taking into account that the ratio between $u$ and $y$ is the exponential of a polynomial, we obtain that $(P(x),Q(x))\in \l_{2n,d}$ if and only if $R(x)\in \l_{2n,d} \cap \k_{2n}$. Therefore, the analysis of polynomial-hyperexponential solutions of a given polynomial degree can be restricted to the trace free family $\mathbb K_{2n}$.

Let us write $R(x) = \sum_{j=0}^{2n} r_jx^j$. Equation \eqref{eq:2} can be reduced to the case of monic polynomial coefficient by the change of variables $x\mapsto\sqrt[2n+2]{\frac{1}{r_{2n}}}x$ which lead us to the equation
\begin{equation}
y''=\left(x^{2n}+\sum\limits_{j=0}^{2n-1}\frac{r_{j}}{\sqrt[2n+2]{\frac{1}{r_{2n}}}}x^{j}\right)y,\hspace{4pt}a_{k}\in\c^{*},a_{i}\in\c.
\end{equation}
For the next step, let us consider $\m_{2n}\subset \k_{2n}$ the family of equations \eqref{eq:2} with monic polynomial coefficient. It is an an algebraic variety isomorphic to $\mathbb C^{2n}$. Since the $(2n+2)$-th root of $r_{2n}$ is an algebraic multivalued function of $r_{2n}$, any equation in $\mathbb K_{2n}$ has $2n+2$ different equivalent reductions in $\mathbb M_{2n}$. This can be seen as an algebraic correspondence $\mathcal C_{2n} \subset \k_{2n} \times \m_{2n}$. This algebraic correspondence is a $(2n+2)$-fold covering space of $\k_{2n}$ by the first projection, $\pi_1$ and the $(2n+2)$ monic reductions of the equation of coefficient $R(x)$ are given by $\pi_2(\pi_1^{-1}(\{R(x)\}))$. Note that $R(x)$ is in $\mathbb L_{2n,d}$ if and only if so are any of its monic reductions. Therefore, if suffices to focus our analysis to equations in the family $\m_{2n}$.

\subsection{Kovacic's algorithm and adapted coordinates in ${\m}_{2n}$}

From now on let $\l'_{2n} = \l_{2n}\cap \m_{2n}$ be the \emph{reduced spectral set} consisting of equations in $\m_{2n}$ having a Liouvillian solution, and let $\l'_{2n,d} = \l_{2n,d} \cap \m_{2n}$ be the \emph{reduced spectral variety} consisting of equations in $\m_{2n}$ having a polynomial-hyperexponential solution of polynomial degree $d$. 

Note that, since D'Alembert reduction does not affect the polynomial degree of polynomial-hyperexponential solutions then a differential equation in the family \eqref{eq:1} has a polynomial-hyperexponential solution of polynomial degree $d$ if and only if so has any of its monic D'Alembert reductions. Therefore, if $\l_{2n,d}$ is a subvariety of $\p_{2n}$ then ${\rm codim}(\l_{2n,d}, \p_{2n}) = {\rm codim}(\l_{2n,d}',\m_{2n})$.

Here we will analyze the existence of Liovillian solutions of equations in the family $\mathbb M_{2n}$. This is done in terms of some known theoretical results obtained by application of Kovacic's algorithm \cite{Kov86}. A first step is to introduce a system of coordinates in $\m_{2n}$ that fits our analysis of equation \eqref{eq:2} better than the coefficients of $R(x)$. 

\begin{lemma}\citep[Lemma 2.4, pp. 275]{AB08}
\label{comsqr} Every monic polynomial $M(x)$ of even degree $2n$ can be written in one only way completing squares, that is,
\begin{equation}
M(x)=A(x)^{2}+B(x),
\end{equation}
with $A(x)= x^n +\sum_{j=0}^{n-1} a_jx^j$ is a monic polynomial of degree $n$ and 
$B(x) = \sum_{j=0}^{n-1} b_jx^j$ is a polynomial of degree at most $n-1$.
\end{lemma}

The map $\m_{2n} \to \mathbb C^{2n}$, $R(x)\mapsto (a_0,\ldots,a_{n-1},b_0,\ldots,b_{n-1})$ is a regular invertible polynomial map. Therefore, we may consider the coefficients of $A(x)$ and $B(x)$ as a system of regular coordinates is $\mathbb M_{2n}$. 
The following results gives us precise information about the sets $\l'_{2n}$ and $\l'_{2n,d}$.

\begin{theorem}\citep[Theorem 2.5, pp. 276]{AB08}\label{polychar} Let us consider the differential equation,
\begin{equation}\label{eq:3}
y'' = M(x)y,
\end{equation}
with $M(x)\in \c[x]$ a monic polynomial of degree $k>0$. Then its differential Galois Group $G$ with coefficients in $\c(x)$ falls in one of the following cases:
\begin{enumerate}
\item $G=\sl$ (non-abelian, non-solvable, connected group).
\item $G=\c^{*}\ltimes\c$ (non-abelian, solvable, connected group).
\end{enumerate}
Furthermore, the second case is given if and only if the following conditions holds:
\begin{enumerate}
\item $M(x)$ has even degree $k=2n$, 
\item Writing $M(x) = A(x)^2+B(x)$ as in lemma \ref{comsqr}, the quantity $\pm b_{n-1}-n$ is a non-negative even integer $2d$, $d\in\z_{\geq 0}$.
\item There exist a monic polynomial $P_{d}$ of degree $d$ satisfying at least one of the following differential equations,
\begin{eqnarray}
P''_{d}+2AP'_{d}-(B-A')P_{d}&=0,\label{auxeq:1}\\
P''_{d}-2AP'_{d}-(B+A')P_{d}&=0\label{auxeq:2}.
\end{eqnarray}
\end{enumerate}
In such case, Liouvillian solutions are given by
\begin{eqnarray}
&y_{1}=P_{d}e^{\int A dx},\hspace{4pt}& y_{2}=y_{1}\int\frac{e^{-2\int A dx}}{P_{d}^{2}},\hspace{4pt}or\\
&y_{1}=P_{d}e^{-\int A dx},\hspace{4pt}& y_{2}=y_{1}\int\frac{e^{2\int A dx}}{P_{d}^{2}}.
\end{eqnarray}
\end{theorem}

A careful read of Theorem \ref{polychar} gives us the following.

\begin{corollary}\label{cor:polychar}
The sets $\l'_{2n}$ and $\l'_{2n,d}$ in $\m_{2n}$ satisfy the following.
\begin{enumerate}
    \item $\l'_{2n} = \bigcup_{d=0}^\infty \l'_{2n,d}$.
    \item $\l_{2n,d}$ is contained in the hypersurface of $\m_{2n}$ of equation $b_{n-1}^2 - (n+2d)^2 = 0$.
\end{enumerate}
Therefore, the sets $\l_{2n}$ and $\l_{2n,d}$ in $\p_{2n}$ satisfy $\l_{2n} = \bigcup_{d=0}^\infty \l_{2n,d}$.
\end{corollary}

\begin{proof}
1. It is a consequence of the dichotomy of the Galois group. In case the group is not $SL_2(\c)$ it leads to a polynomial-hyperexponential solution.
2. It is a direct consequence of point $2$ in the second part of Theorem \ref{polychar},
The last statement of the corollary is a consequence of the point 1. and the fact the the reductions process from $\mathbb P_{2n}$ to $\mathbb M_{2n}$ preserves polynomial-hyperexponential solutions.
\end{proof}

%\noindent We can observe that the solutions $y_{1}$ have a polynomial factor of degree $d$, and an asymptotic factor which is the exponential of a quadrature. We will call this solutions, Liovillian solutions of type $d$ or simply type $d$ solutions.
%\mbox{}\\
%Let us consider $\mathbb{V}_{2n}'$ the connections $y''=(A(x)^{2}+b_{n-1}x^{n-1}+B(x))y$ inside $\mathbb{V}_{2n}$ the connections with even degree coefficient $y'=P_{2n}y$, $P_{2n}\in\c^{*}\times\c^{2n}$. According to theorem \ref{polychar} in the hypersurface generated by the equations  $\pm b_{n-1}=2d+n$ we have the subvariety $\mathcal{L}_{2n,d}''$ set of connections in $\mathbb{V}_{2n}'$ that have a type $d$ solution. If we take into consideration the projections $\pi_{1}$ and $\pi_{2}$ over $\mathbb{V}_{2n}\times\mathbb{V}_{2n}'$ we have that $\mathcal{L}_{2n,d}''$ the connections in $\mathbb{V}_{2n}$ that have a solution of type $d$ are determined by the image of $\pi_{1}(\pi_{2}^{-1}(\mathcal{L}_{2n,d}''))$.

\subsection{Canonical equation}

The following example:
\begin{equation}\label{eqcano}
y''=(x^{2n}+\mu x^{n-1})y, \quad {\mu \in \mathbb C}.
\end{equation}
that we refer to as \emph{canonical equation} gives us some information about the non emptiness of the sets ${\rm L}_{2n,d}$ for large $d$. Due to theorem \ref{polychar}, if \eqref{eqcano} has a Liouvillian solution, the parameter $\mu$ in the canonical coefficient $x^{2n}+\mu x^{n-1}$ is forced to be a discrete parameter that can be $\mu=2d+n$ or either $\mu=-2d-n$, where $d$ is a non-negative integer, which lead us to deal with two different equations,
\begin{eqnarray}
y'' &=&(x^{2n}+(2d+n)x^{n-1})y,\label{eqcano:1}\hspace{4pt}or\\
y'' &=&(x^{2n}-(2d+n)x^{n-1})y.\label{eqcano:2}
\end{eqnarray}
\begin{proposition}
The differential equation \eqref{eqcano:1} is integrable in the liouvillian sense if and only if, $d=(n+1)k$ or $d=(n+1)k+1$ where $k$ is a non-negative integer.
\end{proposition}
\noindent\textit{Proof:} The differential equation \eqref{eqcano:1}, is transformed into the Whittaker differential equation,
 \begin{equation}
 \mathcal{W}''=\left(\frac{1}{4}-\frac{\frac{-2d-n}{2n+2}}{z}+\frac{4(\frac{1}{2n+2})^{2}-1}{4z^{2}}\right)\mathcal{W},
\end{equation}  
through the change of variables $z=\frac{2}{n+1}x^{n+1}$, $y=z^{-\frac{n}{2n+2}}\mathcal{W}$. Applying Martinet-Ramis theorem we have that 
\begin{equation*}
\pm\frac{-2d-n}{2n+2}\pm\frac{1}{2n+2}=\frac{1}{2}+k,\hspace{4pt}k\in\z_{\geq 0},
\end{equation*}
which left only two posibilities, $d=(n+1)k$ or $d=(n+1)k+1$.\qed

\mbox{}\\
\noindent It is easy to see that the change of variables made in above proof also transform the equation \eqref{eqcano:2} into a Whittaker equation. Nevertheless this new equation will have parameters $\kappa=\frac{2d+n}{2n+2}$ and $\mu=\frac{1}{2n+2}$. So via Martinet-Ramis theorem, see \cite{MR89}, we can enunciate the following result analogous to the previous proposition. 
\begin{proposition}\label{pr:w_int}
The differential equation \eqref{eqcano:2} is integrable in the liouvillian sense if and only if, $d=(n+1)k$ or $d=(n+1)k+1$ where $k$ is a non-negative integer.
\end{proposition}
\noindent Moreover, the solutions to the equation \eqref{eqcano} can be explicitly written as
\begin{equation}
\begin{split}
y_{d,n}(x)&=P_{d,n}(x)e^{\frac{x^{n+1}}{n+1}},\hspace{4pt}\textit{if}\hspace{4pt}\mu=2d+n,\hspace{4pt}\textit{or}\\
y_{d,n}(x)&=P_{d,n}(x)e^{-\frac{x^{n+1}}{n+1}},\hspace{4pt}\textit{if}\hspace{4pt}\mu=-(2d+n),
\end{split}
\end{equation}
where the polynomials $P_{d,n}$ can be  find by a Frobenius-like method. Having said that, it is a tedious process. In any case, for $d=(n+1)k$ we have that
\begin{equation}
    P_{d,n}(x)=x^{d}+\sum_{j=n+1}^{d}\zeta_{j}x^{d-j}, \hspace{4pt}\textit{where}\hspace{4pt}\zeta_{j}=0 \hspace{4pt}\textit{for}\hspace{4pt}j\neq (n+1)m.
\end{equation}

On the other hand, for  $d=(n+1)k+1$ 

\begin{equation}
    P_{d,n}(x)=x^{d}+\sum_{j=n+2}^{d}\zeta_{j}x^{d+1-j}, \hspace{4pt}\textit{where}\hspace{4pt}\zeta_{j}=0 \hspace{4pt}\textit{for}\hspace{4pt}j\neq (n+1)m+1.
\end{equation}

\begin{table}[ht]
\centering
\begin{tabular}{|l|l|}
\hline
 $\zeta_{j}$& $j=(n+1)m$, $j=(n+1)m+1$  \\ \hline
 $\mu=2d+2$& $\prod_{r=1}^{m}-\frac{(d+2-r(n+1))(d+3-r(n+1))}{-2(d+2-r(n+1)-n)-2d}$   \\ \hline
 $\mu=-2d-2$& $\prod_{r=1}^{m}-\frac{(d+2-r(n+1))(d+3-r(n+1))}{2(d+2-r(n+1)-n)+2d}$ \\ \hline
\end{tabular}
\end{table}

\begin{corollary}\label{cor:no_empty}
For any pair $(n,d)$ of degrees with $d\equiv 0$ or $d\equiv 1$ mod $(n+1)$ there exist  a monic polynomial $M(x)$ of degree $2n$ such that the equation
\begin{equation}
    y''=M(x)y
\end{equation}
has a polynomial-hyperexponential solution of exponential degree $n+1$ and polynomial degree $d$; therefore  $\mathbb L_{2n,d}$ is non-empty.
\end{corollary}

\section{Analysis of auxiliary equations}\label{auxi}

We refer to equations  \eqref{auxeq:1} and \eqref{auxeq:2} as \emph{auxiliary equations} for equation \eqref{eq:3}. As it is stated in Theorem \ref{polychar} the existence of a Liouvillian solution of equation \eqref{eq:3} depends of the existence of a polynomial solution of the auxiliary equations. In what follows we will show that conditions for the existence of a polynomial solution $P_d$ of given degree is algebraic in the coefficients of $A(x)$ and $B(x)$, and therefore in the coefficients of $M(x)$.

\subsection{Asymptotic iteration method}

The asymptotic iteration method or AIM was introduced by H. Ciftci \textit{et al} in \citep{CHS} as a tool to solve homogeneous differential equations of the form 
\begin{equation}\label{miapreli}
y''=\ell_{0}y'+r_{0}y
\end{equation}
where $\ell_{0}$ and $r_{0}$ are smooth functions defined on a real interval. Nevertheless, the method is purely differential algebraic, so we can extent the result to differential rings of characteristic zero. By derivation of equation \eqref{miapreli} we obtain a sequence of differential equations,
\begin{equation}\label{miaseq}
y^{(j+2)}=\ell_{j}y'+r_{j}y
\end{equation}
where the sequences $\{\ell_j\}_{j\in\mathbb N}$ and $\{r_j\}_{j\in \mathbb N}$ are defined by the recurrence,
\begin{equation}\label{recurrence_lr}
\ell_{j+1}=\ell_{j}'+r_{j}+\ell_{0}\ell_{j}, \quad r_{j+1}=r_{j}'+r_{0}\ell_{j}.    
\end{equation}
and the sequence of obstructions,
$$\delta_j = r_j\ell_{j-1} - \ell_jr_{j-1}.$$
We say that the AIM \emph{stabilizes} at $p>0$ if $\delta_p = 0$. The following statement is a differential algebraic translation of \citep[Theorem 1]{SHC}.

\begin{theorem}\label{MIAExtension}
Let $\ell_{0}$ and $r_{0}$ be elements of a differential field $R$ of characteristic zero. If there exist $p>0$ such that 
\begin{equation}\label{eq:MIAalpha}
\frac{r_{p}}{\ell_{p}}=\frac{r_{p-1}}{\ell_{p-1}}:=\a,
\end{equation}
then differential equation \eqref{miapreli} has general solution,
\begin{equation}
y=u^{-1}(c_{2}+c_{1}\beta),\quad c_1,c_2\mbox{ arbitrary constants,}
\end{equation}
in the extension $R\langle\a, u, u^{-1} , v, \beta\rangle$ where $u, v,\beta$ are solutions of $u'=\a u$, $v'=\ell_{0}v$ y $\beta'=u^{2}v$, respectively,
\end{theorem}

\begin{proof}
By derivation of equation \eqref{miapreli} we obtain,
$$y^{(p+2)} = \ell_py' + r_py,$$
and from there
$$\log(y^{(p+1)})' = 
\frac{\ell_p\left(y' + \frac{r_p}{\ell_p} y\right)}{\ell_{p-1}\left(y' + \frac{r_{p-1}}{\ell_{p-1}} y \right)}.$$
If condition \eqref{eq:MIAalpha} is satisfied, then we have
$$(y^{(p+1)})' = 
\frac{\ell_p}{\ell_{p-1}}y^{(p+1)}.$$
On the other hand, from the recurrence,
we have,
$$\frac{\ell_p}{\ell_{p-1}} = \log(\ell_{p-1})' + \alpha + \ell_0,$$
and replacing into the above equation we obtain,
$$(y^{(p+1)})' = (\log(\ell_{p-1})' + \alpha + \ell_0)y^{(p+1)}.$$
We have that $y^{(p+1)} = c_1\ell_{p-1}uv$ is a general solution for this equation and finally we obtain
$$y' + \alpha y = c_1uv$$
that yields the general solution of the statement.
\end{proof}

The AIM method tests whether the auxiliary equations have polynomial solution. The following statement is a differential algebraic translation of \citep[Theorem 2]{SHC}. There is no difference in the proof, so we refer the reader to the original text.

\begin{theorem}\label{ThMiaPoly} Let $\ell_{0},r_{0}$ be elements in a differential field $R$ of characteristic zero that contains $\c[x]$. 
\begin{enumerate}
\item[(i)] If \eqref{miapreli} has a polynomial solution of degree $p$, then $\delta_p = 0$
\item[(ii)] If $\ell_{p}\ell_{p-1}\neq 0$ and $\d_{p}=0$, then the differential equation \eqref{miapreli} has a polynomial solution of degree at most $p$.
\end{enumerate}
\end{theorem}

\subsection{Liouvillian solutions by means of AIM}

Let us proceed to the AIM of auxiliary equations \eqref{auxeq:1} and \eqref{auxeq:2}. For equation \eqref{auxeq:1} we should start with $\ell_0^+ = -2A(x)$ and $r_0^+ = B(x) - A'(x)$. By the recurrence
law \eqref{recurrence_lr} we have a sequence: 
$$ 
\left[\begin{matrix} \ell^+_{p+1} \\  r^+_{p+1} \end{matrix}\right] = 
\left[\begin{matrix} \ell^+_p \\  r^+_p \end{matrix}\right]' + 
\left[\begin{matrix} -2A(x) & 1  \\  B(x)-A'(x) & 0 \end{matrix}\right]\left[\begin{matrix} \ell^+_p \\  r^+_p \end{matrix}\right]
$$
A condition for the existence of a polynomial solution of degree at most $p$ of \eqref{auxeq:1} is the vanishing of the polynomial $\delta_p^+ = r^+_p\ell^+_{p-1} - r_{p-1}^+\ell^+_{p}$. We proceed analogously with equation \eqref{auxeq:2} obtaining sequences of polynomials $r^-_p$, $\ell^-_p$ and $\delta^-_p$.

In order to model this process, let us consider $\mathbb Q\{a,b\}$ the ring of differential polynomials in two differential variables $a,b$. We may consider the following $\mathbb Q$-linear differential operator in the space of $2$ by $2$ matrices.
$$\varphi\colon {\rm Mat}_{2\times 2}(\mathbb Q\{a,b\}) \to {\rm Mat}_{2\times 2}(\mathbb Q\{a,b\}),
\quad C \mapsto \varphi(C) =  C' + \left[\begin{matrix} -2a & 1  \\  b-a' & 0 \end{matrix}\right]C$$
We consider the iterations of this map. If we give to the differential variables $a$, $b$ the values of the polynomials $A(x)$ and $B(x)$ we obtain:
$$(\varphi^{p+1}({\rm Id}))(A(x),B(x)) =  \left[\begin{matrix} \ell_p & \ell_{p-1}  \\  r_p & r_{p-1} \end{matrix}\right].$$
Let us define the sequence of universal differential polynomials,
$$\Delta_p = - \det(\varphi^{p+1}({\rm Id})) \in \mathbb Q\{a,b\}.$$

\begin{table}[h]
\centering
\begin{tabular}{|l|l|}
\hline
$\Delta_0$ &  $b-a'$\\ \hline
$\Delta_1$ & $2 a \left(a''-b'\right)+4 b a'-3 (a')^2-b^2$ \\ \hline
$\Delta_2$ &  
$\begin{aligned}
   & -9 b^2 a'-15 (a')^3-2 (-3 a'' b'+2 (a^2 (a^{(3)}-b'') +(a'')^2) +(b')^2) \\
   &+b (-a^{(3)}+2 a (3 b'-5 a'')+23 (a')^2+b'')+a' (3 a^{(3)}-2 a (7 b'-9 a'')-3 b'')+b^3

\end{aligned}$
 
\\ \hline
$\Delta_3$ & 
$\begin{aligned}
&2 (2 a''+2 a (b-4 a')+4 a^3-b') (a^{(4)}-4 a^2 (b'-a'')+b (10 a''-4 b')+10 a' b'+8 a^3 (b-a')\\
&+2 a (-a^{(3)}-14 b a'+12 (a')^2+b''+2 b^2)-16 a' a''-b^{(3)})-(-5 a^{(3)} +a (26 a''-10 b')\\
&+12 a^2 (b-5 a')-10 b a'+21 (a')^2 +16 a^4+3 b''+b^2) (-a^{(3)}-2 a (b'-a'')+4 a^2 (b-a')\\
&-6 b (a')+5 a'^2+b''+b^2)
\end{aligned}$
\\ \hline
\end{tabular}
\caption{First values of the universal differential polynomials $\Delta_n$.}
\end{table}

As we will see this sequence $\{\Delta_d\}_{d\in\mathbb N}$ of differential polynomials governs the Liouvillian integrability of equation \eqref{eq:3} for any even degree $2n$ of $M(x)$.

\begin{theorem}\label{seqdelta}
%The sequence $\Delta_0$, $\Delta_1$, $\Delta_2$, $\ldots$
Equation \eqref{eq:3} with $M(x) = A(x)^2 + B(x)$ has a polynomial-hyperexponential solution of polynomial degree $d$ if and only if,
$$b_{n-1}^2 = (n +2d)^2 \quad\mbox{and} \quad \Delta_{d}(A(x),B(x))\Delta_d(-A(x),B(x))=0.$$
Therefore $\mathbb L'_{2n,d}$ is an algebraic subvariety of $\mathbb M_{2n}$ contained in the union of the irreducible hypersurfaces of equations:
\begin{equation*}
b_{n-1}=2d+n,\quad -b_{n-1}=2d+n.
\end{equation*}
\end{theorem}

\begin{proof}
Note that, by definition of the sequence $\Delta_d$ we have $\delta^+_d = \Delta_d(A(x),B(x))$. Analogously, the application of the AIM to equation \eqref{auxeq:2} produces an obstruction, $\delta^-_d$. Note that, because of the symmetry between equations \eqref{auxeq:1} and \eqref{auxeq:2} $\delta^-_d = \Delta_d(-A(x),B(x))$.

We need only to check that $\ell_{d-1}^\pm\ell_{d}^\pm\neq 0$ for the auxiliar equations. This comes easily from the fact that $\ell_0^\pm = \pm 2A(x)$ is of bigger degree than $r_0 = B\pm A'$. Note that $\Delta_d(A(x),B(x))\Delta_d(-A(x),B(x))$ is a polynomial in $x,a_0,\ldots,a_{n-1},b_0,\ldots,b_{n-1}$. Its coefficients as a polynomial in $x$ are the algebraic equations of the restricted spectral variety $\mathbb L'_{2n,d}$ in $\mathbb M_{2n}$.
\end{proof}

\begin{example}\label{ex:quadratic}
As a first example of AIM applications, let us consider an equation on $\mathbb{M}_{2}$
\begin{equation}
    y''= ((x+a_{0})^2+ b_0)y.
\end{equation}
An elementary traslation as $x\mapsto x+a_0$ reduces the determination of $\mathbb{L'}_2$ structure to an analysis of liouvillian-integrability conditions for quantum harmonic oscillator
\begin{equation}
    y''=(x^2 + b_0)y.
\end{equation}
These conditions are $b_0^2 = (2d+1)^2$ and $\Delta_d (x,b_0)\Delta_d (-x,b_0)=0$. It is easy to verify that $\Delta_d (x,b_0)=\Delta_d (-x,b_0)=2^{d+1}\prod_{k=0}^{d} d-k$. Therefore,
\begin{equation}
    \mathbb{L'}_{2,d} \,\, :\,\, (b_0 + 2d+1)(b_0 - 2d-1)=0
\end{equation}
\end{example}

Let us note that for a given equation $\m_{2n}$, conditions $b_{n+1} = 2d + n$ and $b_{n+1} = -2d - n$ are mutually exclusive. In the first case, auxiliary equation \eqref{auxeq:1} may have a polynomial solution but not \eqref{auxeq:2}. The opposite occurs in the second case. Therefore, we decompose the spectral variety $\l_{2n,d}'$ as the disjoint union of two components $\mathbb{L'}_{2n,d}=\mathbb{L}_{2n,d}^{+} \cup \mathbb{L}_{2n,d}^{-}$. The first component $\mathbb{L}_{2n,d}^{+}$ correspond to equations whose auxiliary equation \eqref{auxeq:1} has a polynomial solution of degree $d$ and the second component $\mathbb{L}_{2n,d}^{-}$ correspond to equations whose  auxiliary equation \eqref{auxeq:2} has a polynomial solution of degree $d$.

\begin{definition}
$\mathbb{L'}_{2n,d}=\mathbb{L}_{2n,d}^{+} \cup \mathbb{L}_{2n,d}^{-}$, where
\begin{equation*}
    \mathbb{L}_{2n,d}^{+}=\begin{cases}
    b_{n-1}=2d+n\\
    \Delta_d(A(x),B(x))=0,
    \end{cases}
\end{equation*}
and
\begin{equation*}
    \mathbb{L}_{2n,d}^{-}=\begin{cases}
    b_{n-1}=-2d-n\\
    \Delta_d(-A(x),B(x))=0.
    \end{cases}
\end{equation*}
\end{definition}

As in the previous Example \ref{ex:quadratic} it is always possible to get rid of the coefficient $a_{n-1}$ by means of a translation in the $x$ axis. Therefore is convenient to consider the sets,
 $$V_{2n,d}^{\pm}=\{a_{n-1}=0\}\cap\mathbb{L}_{2n,d}^{\pm}.$$ 
whose equations are easier to describe. For instance, in $\mathbb{M}_4$ and $\m_6$ we restrict our analysis to equations of the forms:
\begin{equation}
    y''=((x^2 + a_0)^2 +b_1 x + b_0)y.
\end{equation}
and
\begin{equation}
    y''=((x^3 +a_1 x + a_0)^2 +b_2 x^2 + b_1 x + b_0)y
\end{equation}
respectively. The following calculations of the equations of $V^+_{2n,d}$ for $n=2,3$ and small values of $d$, in Tables \ref{t:V4} and \ref{t:V6} is performed by means of the universal differential polynomials $\Delta_d$.

\begin{table}[ht]
\centering
\begin{tabular}{|l|l|lll}
\cline{1-2}
$V_{4,0}^{+}$ & $\begin{cases}
b_1=2 \\
b_0=0
\end{cases}$&  &  &  \\ \cline{1-2}
$V_{4,1}^{+}$ & $\begin{cases}
b_1= 4\\
b_0^2 +4a_0=0
\end{cases}$ &  &  &  \\ \cline{1-2}
$V_{4,2}^{+}$ & $\begin{cases}
b_1= 6 \\
b_0^3+16 a_0 b_0-16=0
\end{cases}$ &  &  &  \\ \cline{1-2}
$V_{4,3}^{+}$ & $\begin{cases}
b_1= 8\\
b_0^4 + 40a_0 b_0^2-96b_0+144a_0^2=0
\end{cases}$ &  &  &  \\ \cline{1-2}
$V_{4,4}^{+}$ & $\begin{cases}
b_1= 10\\
b_0^5+80a_0 b_0^3-336b_0^2 +1024a_0^2 b_0 -3072a_0=0
\end{cases}$ &  &  &  \\ \cline{1-2}
$V_{4,5}^{+}$ & $\begin{cases}
b_1= 12\\
b_0^6-140a_0 b_0^4+896b_0^3 -4144a_0^2 b_0^2 +28160a_0 b_0-14400a_0^3-25600=0
\end{cases}$ &  &  &  \\ \cline{1-2}
$V_{4,6}^{+}$ & $\begin{cases}
b_1= 14\\
\begin{aligned}
b_0^7+224a_0 b_0^5-2016b_0^4 +12544a_0^2 b_0^3 -142848a_0 b_0^2+147456a_0^3 b_0&+288000b_0\\
&-884736a_0^2=0
\end{aligned}
\end{cases}$ &  &  &  \\ \cline{1-2}
\end{tabular}
\caption{Algebraic equations of restricted spectral varieties $V^+_{4,d} = \l^+_{4,d} \cap \{a_1 = 0\}$
for small values of $d$.
}\label{t:V4}
\end{table}

\begin{table}[ht]
\centering
\begin{tabular}{|l|l|}
\hline
$V_{6,0}^{+}$  & $\begin{cases}b_2= 3\\
b_1=0\\
b_0-a_1=0
\end{cases}$\\ \hline
$V_{6,1}^{+}$  &  
$\begin{cases}b_2= 5\\
2 a_1 b_1-8 a_0-2 b_0 b_1=0\\
-6 a_1-b_1^2+2 b_0=0\\
4 a_1 b_0-2 a_0 b_1-3 a_1^2-b_0^2=0
\end{cases}$\\ \hline
$V_{6,2}^{+}$  &  $\begin{cases}b_2= 7\\
23 a_1^2 b_0-9 a_1 b_0^2-14 a_0 a_1 b_1+6 a_0 b_0 b_1-15 a_1^3-24 a_1+32 a_0^2+b_0^3-2 b_1^2+8 b_0=0\\
9 a_1^2 b_1-12 a_1 b_0 b_1+6 a_0 b_1^2+24 a_0 b_0-24 a_0 a_1+3 b_0^2 b_1-12 b_1=0\\
-3 a_1 b_1^2+36 a_1 b_0+24 a_0 b_1-30 a_1^2-6 b_0^2+3 b_0 b_1^2-48=0\\
22 a_1 b_1-32 a_0+b_1^3-6 b_0 b_1=0
\end{cases}$\\ \hline
$V_{6,3}^{+}$  &  $\begin{cases}b_2= 9\\
\begin{aligned}
&176 a_1^3 b_0-86 a_1^2 b_0^2-116 a_0 a_1^2 b_1+16 a_1 b_0^3-20 a_1 b_1^2+264 a_1 b_0+80 a_0 a_1 b_0 b_1\\
&-12 a_0^2 b_1^2-144 a_0^2 b_0
-12 a_0 b_0^2 b_1+120 a_0 b_1-105 a_1^4-372 a_1^2+432 a_0^2 a_1\\
&-b_0^4-36 b_0^2+8 b_0 b_1^2-288=0
\end{aligned}\\
\begin{aligned}
&60 a_1^3 b_1-92 a_1^2 b_0 b_1+56 a_0 a_1 b_1^2+192 a_0 a_1 b_0+36 a_1 b_0^2 b_1+96 a_1 b_1-48 a_0 b_0^2\\
&-24 a_0 b_0 b_1^2-288 a_0^2 b_1-144 a_0 a_1^2+576 a_0+8 b_1^3-4 b_0^3 b_1=0
\end{aligned}\\
\begin{aligned}
&-18 a_1^2 b_1^2+372 a_1^2 b_0-132 a_1 b_0^2+24 a_1 b_0 b_1^2-72 a_0 a_1 b_1-12 a_0 b_1^3-72 a_0 b_0 b_1\\
&-252 a_1^3-288 a_1+12 b_0^3-6 b_0^2 b_1^2+48 b_1^2+288 b_0=0
\end{aligned}\\
\begin{aligned}
&4 a_1 b_1^3-48 a_0 b_1^2+136 a_1^2 b_1-160 a_1 b_0 b_1+288 a_0 b_0-864 a_0 a_1-4 b_0 b_1^3+24 b_0^2 b_1\\
&+192 b_1=0
\end{aligned}\\
-52 a_1 b_1^2+144 a_0 b_1+120 a_1 b_0-252 a_1^2-b_1^4+12 b_0 b_1^2-12 b_0^2-288=0
\end{cases}$\\ \hline
\end{tabular}
\caption{Algebraic equations of restricted spectral varieties $V^+_{6,d} =  \l^+_{6,d}\cap \{a_1 = 0\}$
for small values of $d$.
}\label{t:V6}
\end{table}

\subsection{Codimension of the spectral variety}\label{spcvar}

As the degree in $x$ of the polynomials $\Delta_p(A(x),B(x))$ grows quickly with $p$ and the degree of $M(x) = A(x)^2 + B(x)$ it seems that the sets $\mathbb L_{2n,p}$ are smaller as $p$ grows. However, a direct analysis of the auxiliary equations allows us to bound the codimension of the spectral varieties $\mathbb{L}_{2n,d}$ in $\m_{2n}$. As we have seen before the algebraic equations for $\mathbb{L}_{2n,0}$ are well expressed by the obstruction $\Delta_{0}(a,b)=b-a'$, so henceforth we will consider $d>0$.

\begin{proposition}\label{pr:codim}
If $\mathbb{L'}_{2n,d}$ is not empty, then ${\rm codim}(\mathbb{L'}_{2n,d},\m_{2n})\leq n$.
\end{proposition}

\begin{proof}
Now, let us suppose that $P_{d}=\sum_{k=0}^{d}p_{k}x^k$ is a solution to one of the following algebraic equations

\begin{equation}
P_{d}''\pm 2AP_{d}'-(B\mp A')P_{d}=0
\end{equation}
where $A=x^n + \sum_{k=1}^{n}a_{n-k}x^{n-k}$ and $B=\sum_{k=1}^{n}b_{n-k}x^{n-k}$. Hence the coefficients of the  polynomial 

\begin{equation}
\begin{split}
&\sum_{k=2}^{d}k(k-1)p_{k}x^{k-2}\pm \left(2x^n + \sum_{k=1}^{n}2a_{n-k}x^{n-k}\right)\left(\sum_{k=1}^{d}kp_{k}x^{k-1}\right)\\
-&\left(\sum_{k=1}^{n}b_{n-k}x^{n-k}\mp nx^{n-1} + \sum_{k=1}^{n-1}(n-k)a_{n-k}x^{n-1-k}\right)\left(\sum_{k=0}^{d}p_{k}x^k\right)=0
\end{split}
\end{equation}
in $\mathbb{C}[x]$ vanish. This give place to a system of equations which are sufficient conditions for the existence of $P_{d}$,
\begin{equation}\label{syscond}
\begin{bmatrix}
a_1 -b_0 & 2a_2 & 2 & 0 & \cdots &  0& 0\\ 
2a_2 -b_1 & 3a_1 -b_0 & 4a_0   & 6 & \cdots &0  &0 \\ 
 &  &  &  & \ddots  &\ast & \ast\\ 
0 & 0 & 0 & 0 & \cdots & 2(d-1)+n-b_{n-1} & (2d+n-1)a_{n-1}-b_{n-2}\\ 
0 & 0 & 0 & 0 & \cdots & 0 &2d+n-b_{n-1} \\ 
\end{bmatrix}
\begin{bmatrix}
p_0\\ 
p_1\\ 
\vdots\\ 
p_d
\end{bmatrix}=0.
\end{equation} 

We will denote  the coefficient matrix of the system \eqref{syscond} by $M_{d,n}^{\pm}(A,B)$. Note  this matrix has size $(d+n)\times(d+1)$ and it also has the property 
\begin{equation}
M_{d,n}^{\pm}(A,B)=\left[\begin{array}{ccc|c}
 & & & 0\\
 &M_{d-1,n}^{\pm}(A,B) & &\vdots \\
 & && \ast \\ \hline
 0 &\cdots &0 &2d+n\pm b_{n-1} \\
\end{array}\right].
\end{equation}
\begin{remark}
As there is no solution $P$ of degree less than $d$, then $rk(M_{d-1,n}^{\pm}(A,B))=d$.
\end{remark}

In order to determinate the codimension of $\mathbb{L}_{2n,d}'$ around a point $(A_0 ,B_0 )$ we shall choose a suitable $d\times d$ submatrix $D$ of $M_{d-1,n}^{\pm}(A_0,B_0)$ such that its determinant is different from zero. In addition, the vanishing of the determinants of the matrices set by adding one of the remaining $n$ rows of $M_{d-1,n}^{\pm}(A_0,B_0)$ to $D$, generates $n$ conditional equations which guarantees the existence of a non-trivial solution to \eqref{syscond}. 
\end{proof}

\subsection{An example: case $n=3$}\label{furt}
As an useful example in order to illustrate further computes, specially for looking accurate spectral values on Schr\"odinger type problems, let us assume that $A(x)=x^{3}+a_{3,1}x^{2}+a_{3,2}x+a_{3,3}$ and $B(x)=b_{2,0}x^{2}+b_{2,1}x+b_{2,2}$. So, the analysis on previous subsection \ref{spcvar} for case $n=3$ can be summarized with the following proposition.

\begin{proposition}\label{exampletesis}
A necessary condition for equation
\begin{equation}\label{examplesextic}
    y''-(2x^{3}+2a_{3,1}x^{2}+2a_{3,2}x+2a_{3,3})y'-((b_{2,0}+3)x^2 + (2a_{3,1}+b_{2,1})x+a_{3,2}+b_{2,2})y=0
\end{equation}
in order to have a polynomial solution of degree $d$ is $b_{2,0}+3=-2d$ for $d=0,1,2,\dots$

\noindent On the other hand, sufficient conditions are coded by the solutions of the linear system associated to the matrix
\begin{equation}
M_{d}^{-}(A,B)=\begin{bmatrix}
\a_{0} &\b_{0}  &\g_{0}  &  &  &  &  &  \\ 
\zeta_{1} &\a_{1}  &\b_{1}  &\g_{1}  &  &  &  &   \\ 
\eta_{2} &\zeta_{2}  &\a_{2}  &\b_{2}  &\g_{2}  &  &  &   \\ 
      &\ddots &\ddots &\ddots &\ddots  & \ddots &  &   \\ 

                    &  &\eta_{d-2}&\zeta_{d-2}&\a_{d-2}  &\b_{d-2}  &\g_{d-2}\\ 
                   &  &       &\eta_{d-1}&\zeta_{d-1}&\a_{d-1}& \b_{d-1}  \\ 
                   &  &       &  &\eta_{d}  &\zeta_{d}  &\a_{d} \\
                          &  &       &  &0  &\eta_{d+1}  &\zeta_{d+1} 
\end{bmatrix} 
\end{equation}
where
\begin{equation*}
\begin{split}
\a_{k}&=-a_{3,2}(2k+1)-b_{2,2},\\
\b_{k}&=-2a_{3,3}(k+1),\\
\g_{k}&=(k+2)(k+1),\\
\zeta_{k}&=-2a_{3,1}k-b_{2,1},\\
\eta_{k}&=-2k-b_{2,0}+1,
\end{split}
\hspace{4pt}
\begin{split}
k=0,1,2,\dots\\ 
k=0,1,2,\dots\\ 
k=0,1,2,\dots\\ 
k=1,2,3\dots\\ 
k=2,3,4\dots
\end{split}
\end{equation*}
It creates a set of at most two polynomial equations in the variables $a_{3,0}$, $a_{3,1}$, $a_{3,2}$, $a_{3,3}$, $a_{2,0}$, $a_{2,1}$, $a_{2,2}$ which guarantees likewise a non-trivial solution to the associated system to $M_{d}^{-}(A,B)$ and a polynomial solution  to equation \eqref{examplesextic}. 
\end{proposition}
 \begin{proof}
 This is a restriction of the  analysis developed on section \ref{spcvar} to $n=3$.
 \end{proof}

Generically we can suppose that our $d\times d$ principal minor is different from zero, so the equations given by proposition \ref{exampletesis} are the determinants
\begin{equation}
\D_{d,3}^{1}(-A,B)=\begin{vmatrix}
\a_{0} &\b_{0}  &\g_{0}  &  &  &  &  &  \\ 
\zeta_{1} &\a_{1}  &\b_{1}  &\g_{1}  &  &  &  &   \\ 
\eta_{2} &\zeta_{2}  &\a_{2}  &\b_{2}  &\g_{2}  &  &  &   \\ 
      &\ddots &\ddots &\ddots &\ddots  & \ddots &  &   \\ 

                    &  &\eta_{d-2}&\zeta_{d-2}&\a_{d-2}  &\b_{d-2}  &\g_{d-2}\\ 
                   &  &       &\eta_{d-1}&\zeta_{d-1}&\a_{d-1}& \b_{d-1}  \\ 
                   &  &       &  &\eta_{d}  &\zeta_{d}  &\a_{d} 
\end{vmatrix} =0
\end{equation}
and
\begin{equation}
\D_{d,3}^{2}(-A,B)=\begin{vmatrix}
\a_{0} &\b_{0}  &\g_{0}  &  &  &  &  &  \\ 
\zeta_{1} &\a_{1}  &\b_{1}  &\g_{1}  &  &  &  &   \\ 
\eta_{2} &\zeta_{2}  &\a_{2}  &\b_{2}  &\g_{2}  &  &  &   \\ 
      &\ddots &\ddots &\ddots &\ddots  & \ddots &  &   \\ 

                    &  &\eta_{d-2}&\zeta_{d-2}&\a_{d-2}  &\b_{d-2}  &\g_{d-2}\\ 
                   &  &       &\eta_{d-1}&\zeta_{d-1}&\a_{d-1}& \b_{d-1}  \\ 
                   &  &       &  &0  &\eta_{d+1}  &\zeta_{d+1} 
\end{vmatrix} =0
\end{equation}
%\begin{remark}
%We can see through this analysis that the number of equations which define  the variety $\mathbb{L}_{2n,d}'$ %is directly related to $n$, so the complexity of these varieties will increase as $n$ does. 
%\end{remark}
Several detailed examples of this equations can be found on \citep{CHSD, MTH}.

\subsection{Proof of Theorem \ref{th:main_structure} }

We can now state the proof, which follows easily from the other results. Statement (a) is a direct consequence of Proposition \ref{pr:w_int}. Statement (b) is a consequence of Theorems \ref{seqdelta} and proposition \ref{pr:codim}. Note that, from d'Alembert reduction, the codimension of $\mathbb L'_{2n,d}$ in $\mathbb M_{2n}$ coincide with that of $\mathbb L_{2n,d}$ in $\mathbb P_{2n}$. Statement (c) and (d) are also clear, as $\mathbb L'_{2n,d}$ is contained in the union of hyperplanes of equations $b_{n+1} = 2d+n$ and $b_{n-1} = -2d-n$. \qed

\section{Schr\"odinger Equation }\label{schr}

Let us consider the one dimensional stationary Shr\"odinger equation,
\begin{equation}\label{eq:sch}
\psi'' = (\lambda - U(x))\psi    
\end{equation}
with a polynomial potential $U(x)$. We say that the potential $U(x)$ is quasi-exactly solvable if there are some values of $\lambda$ for wich equation \eqref{eq:sch}  has a Liouvillian solution. This is equivalent to say that the line, 
$$\{\lambda-U(x) \,\colon\, \lambda\in\mathbb C\}\subseteq \mathbb M_{2n}$$
paremeterized by $\lambda$, intersects the spectral set $\mathbb L_{2n}$.

As it is well know, and we examined in Example \ref{ex:quadratic}, any quadratic potential is quasi-exactly solvable (and more over, exactly solvable). It is also clear that any quasi-exactly solvable potential is of even degree. Let us assume from now on that $U(x)$ is of degree $2n\geq 4$.

We consider the decomposition $-U(x) = A(x)^2+B(x)$ as in Theorem \ref{polychar}. We define the \emph{arithmetic condition} of $U(x)$ as the complex number,
$$d = \frac{|b_{n-1}|-n}{2}$$
where $b_{n-1}$ is the coefficient of $x^{n-1}$ the polynomial $B(x)$ appearing in the unique decomposition $-U(x) = A(x)^2 + B(x)$. Note that a necessary condition for $U(x)$ to be quasi exactly solvable is its arithmetic condition to be a non-negative integer. In such case the intersection between the line:
$$\{\lambda-U(x) \,\colon\, \lambda\in\mathbb C\}\subseteq \mathbb M_{2n}$$
and $\l_{2n}$ is confined to the spectral variety $\l_{2n,d}$. 

Let us consider the universal sequence of differential polynomials $\Delta_d\in \mathbb Q\{a,b\}$ as in Theorem \ref{seqdelta}. The following lemma allows us to bound the number of admissible values of energy (for which the Schr\"odinger equation admits a Liouvillian solution) of any quasi-exactly solvable polynomial potential. Let us make clear that by the degree of a differential polynomial $\Delta_d$ in the variable $b$ we mean its ordinary degree: that is we consider $a,a',a'',\ldots,b,b',b'',\ldots$ as an infinite set of independent variables. 

\begin{lemma}\label{lm:degree}
 The degree of $\Delta_d$ in the variable $b$ is at most $d+1$.
\end{lemma}

\begin{proof}
Let us recall the differential polynomials $\ell_d$ and $r_d$ appearing in the definition of $\Delta_d$. 
Let us prove first:
\begin{itemize}
    \item[(a)] The degree of $\ell_d$ in the variable $b$ is small or equal to $\frac{d+1}{2}$.
    \item[(b)] The degree of $r_d$ in the variable $b$ is small or equal to $\frac{d+2}{2}$.
\end{itemize}
The degree of $\ell_0 = -2a$ in the variable $b$ is $0$ an the degree of $r_0 = b-a'$ in the variable $b$ is $1$. Therefore (a) and (b) hold for $d=0$. Now, from the recurrence law \eqref{recurrence_lr} we have that the degree in $b$ of $\ell_{j+1}$ is at most that of $r_j$ and that the degree in $b$ of $r_{j+1}$ is at most a unit bigger that the degree of $\ell_j$. This proves (a) and (b). The degree of $\delta_d$ is at most the maximum between the sum of the degrees of $\ell_d$ and $r_{d-1}$ and the sum of the degrees of $\ell_{d-1}$ and $r_d$; which is at most $d+1$.
\end{proof}

\begin{theorem}\label{th:eigen}
Let $U(x)$ be an algebraically quasi-solvable polynomial potential, and let $d$ be its arithmetic condition. The number of values of the energy parameter $\lambda$ such that Equation \eqref{eq:sch} has a Liouvillian solution is at most $d+1$.
\end{theorem}

\begin{proof}
Generically, we may consider that $U(x)$ has no independent term. Then the condition on $\lambda$ for the existence of a Liouvillian solution is the vanishing of $\Delta_d(A(x),B(x)+\lambda)$
which is a polynomial in $x$ of $\lambda$. The number of values of $\lambda$ for which this polynomial vanish can not be greater than its degree in $\lambda$. Clearly, the degree in $\lambda$ of $\Delta_d(A(x),B(x)+\lambda)$ can not exceed the degree in $b$ of $\Delta_d(a,b)$ which is bounded by $d+1$ by Lemma \ref{lm:degree}.
\end{proof}

\begin{example}
In order to illustrate the procedures developed here let us consider the non-singular Turbiner potential 
\begin{equation}\label{turbiner}
U(x)=x^6 - (4J+1)x^2
\end{equation} 
where $J$ is a non-negative integer. This potential has been studied in several papers, including \citep{BD96}. Let  $\mathbf{C}_{d}\subset\mathbb{L}'_{6,d}$ be the set consisting of all possible values for $J$ and $\lambda$ with polynomial hyperexponential solutions of polynomial degree $d$. In virtue of Theorem \ref{seqdelta} it is a subvariety of $V(2J-d-1)$. So, $d$ shall only take non-negative odd values.

On the other hand, we can easily compute the equations of $\mathbf{C}_{d}$ through the universal differential polynomial $\Delta_{d}(x^3,-(4J+1)x^2-\lambda)$ for the auxiliary equation
\begin{equation}
P_{d}''-2x^3 P_{d}'-(3x^2 -(4J+1)x^2-\lambda)P_{d}=0.
\end{equation} 
For the case $d=1$ we get the following equations
\begin{equation}
  \begin{cases}2J-2=0\\
  -\lambda ^2=0\\
  2 (-4 J-1) \lambda +12 \lambda=0\\
  -(-4 J-1)^2-8 (-4 J-1)-15=0.
\end{cases}  
\end{equation}
Taking into account above consideration we compute the first seven  equations for $\mathbf{C}_{d}$ 
\begin{table}[ht]
\centering
\begin{tabular}{|l|l|}
\hline
$d$  & $\mathbf{C}_{d}$ \\ \hline
 1  &  $\begin{cases}J=1\\ \lambda=0\end{cases}$\\ \hline
 3  &$\begin{cases}J=2\\ \lambda^2-24=0\end{cases}$  \\ \hline
 5 &$\begin{cases}J=3\\ \lambda^3-128\lambda=0\end{cases}$  \\ \hline
 7  &$\begin{cases}J=4\\ \lambda^4  - 400 \lambda^2 +12096=0 \end{cases}$  \\ \hline
 9 &$\begin{cases}J=5\\ -\lambda ^5+960 \lambda ^3-129024 \lambda=0 \end{cases}$  \\ \hline
11 &$\begin{cases}J=6\\ \lambda ^6-1960 \lambda ^4+729280 \lambda ^2-26611200=0\end{cases}$  \\ \hline
13 &$\begin{cases}J=7\\ -\lambda ^7+3584 \lambda ^5-2934784 \lambda ^3+438829056 \lambda=0\end{cases} $  \\ \hline
\end{tabular}
\caption{Spectral system of Schr\"odinger equation associated to \eqref{turbiner}}
\end{table}
\end{example}

\section*{Final Remarks}
In this paper we developed a technique to obtain Liouvillian solutions for parameterized  second order linear differential equations with polynomial coefficients. In particular case, we study the set of possible values of energy to get Liouvillian solutions of Schr\"odinger equations with anharmonic potentials. We adapted asymptotic iteration method, Kovacic's Algorithm and previous results provided in \cite{AB08,A09,A10,AMW11} in terms of algebraic varieties extending slightly the known results about polynomial quasi-solvable potentials.
\section*{Acknowledgements}

This research has been partially funded by Colciencias project “Estructuras lineales en geometr\'ia y topolog\'ia” 776-2017 code 57708 (Hermes UN 38300). 
\bibliographystyle{plain}
\bibliography{ABV110819}

\end{document}